\documentclass[10pt, onecolumn]{IEEEtran}
\usepackage{amsthm}
\usepackage{amsmath}
\usepackage{amsfonts}
\usepackage{graphicx}
\usepackage{latexsym}
\usepackage{amssymb}
\usepackage{stmaryrd}
\usepackage{comment}
\usepackage{amscd}
\usepackage[small]{caption}
\usepackage[dvipsnames]{xcolor}
\usepackage{algorithmicx}
\usepackage{algorithm}
\usepackage[noend]{algpseudocode}
\algrenewcommand\alglinenumber[1]{\scriptsize #1:}
\algrenewcommand\algorithmicindent{1em}%
\allowdisplaybreaks
\usepackage{graphicx}
\usepackage{subfigure}
\usepackage{wrapfig}
\usepackage{float}

\usepackage{mathtools}
\newcommand{\bea}{\begin{eqnarray}}
\newcommand{\eea}{\end{eqnarray}}
\newcommand{\bean}{\begin{eqnarray*}}
\newcommand{\eean}{\end{eqnarray*}}
\newcommand{\ceil}[1]{\left\lceil #1 \right\rceil}
\newcommand{\floor}[1]{\left\lfloor #1 \right\rfloor}

\newcommand{\sbinom}[2]{\left( \begin{array}{c} #1 \\ #2 \end{array} \right) }

\newcommand{\field}[1]{\mathbb{#1}}

\newcommand{\R}{\field{R}}

\newcommand{\cC}{{\cal C}}
\newcommand{\cD}{{\cal D}}

\newcommand{\cI}{{\cal I}}

\newcommand{\cL}{{\cal L}}

\newcommand{\cR}{{\cal R}}
\newcommand{\cS}{{\cal S}}

\newcommand{\cW}{{\cal W}}
\newcommand{\cX}{{\cal X}}
\newcommand{\cY}{{\cal Y}}

\newcommand{\sG}{\script{G}}

\newcommand{\sP}{\script{P}}

\newcommand{\bfg}{{\boldsymbol g}}

\newcommand{\bfx}{{\boldsymbol x}}
\newcommand{\bfy}{{\boldsymbol y}}
\newcommand{\bfz}{{\boldsymbol z}}

\DeclareMathOperator*{\argmax}{arg\,max}
\DeclareMathOperator*{\argmin}{arg\,min}

\DeclareMathAlphabet{\mathbfsl}{OT1}{cmr}{bx}{it}
\newcommand{\uuu}{\kern-1pt\mathbfsl{u}\kern-0.5pt}
\newcommand{\vvv}{\kern-1pt\mathbfsl{v}\kern-0.5pt}

\newcommand{\myboxplus}{\kern1pt\mbox{\small$\boxplus$}}

\makeatletter \DeclareRobustCommand{\sbinom}{\genfrac[]\z@{}}
\makeatother
\newcommand{\G}[2]{\sbinom{{#1}\kern-1pt}{{#2}\kern-1pt}}
\newcommand{\Gq}[2]{\sbinom{{#1}\kern-0.25pt}{{#2}\kern-0.25pt}}

\newcommand{\Ps}{\smash{{\sP\kern-2.0pt}_q\kern-0.5pt(n)}}
\newcommand{\sPs}{\smash{{\sP\kern-1.5pt}_q(n)}}
\newcommand{\Ptwo}{\smash{{\sP\kern-2.0pt}_2\kern-0.5pt(n)}}
\newcommand{\Ptwom}{\smash{{\sP\kern-2.0pt}_2\kern-0.5pt(m)}}
\newcommand{\Ptwonm}{\smash{{\sP\kern-2.0pt}_2\kern-0.5pt(n+m)}}
\newcommand{\Ptwoa}{\smash{{\sP\kern-2.0pt}_2\kern-0.5pt(1)}}
\newcommand{\Ptwob}{\smash{{\sP\kern-2.0pt}_2\kern-0.5pt(2)}}
\newcommand{\Ptwoc}{\smash{{\sP\kern-2.0pt}_2\kern-0.5pt(3)}}
\newcommand{\Ptwod}{\smash{{\sP\kern-2.0pt}_2\kern-0.5pt(4)}}
\newcommand{\Ptwoe}{\smash{{\sP\kern-2.0pt}_2\kern-0.5pt(5)}}
\newcommand{\Ptwof}{\smash{{\sP\kern-2.0pt}_2\kern-0.5pt(6)}}
\newcommand{\Ptwokm}{\smash{{\sP\kern-2.0pt}_2\kern-0.5pt(2k-1)}}
\newcommand{\Pone}{\smash{{\sP\kern-2.5pt}_2\kern-0.5pt(n{-}1)}}

\newcommand{\Gr}{\smash{{\sG\kern-1.5pt}_q\kern-0.5pt(n,k)}}
\newcommand{\Gi}{\smash{{\sG\kern-1.5pt}_q\kern-0.5pt(n,i)}}
\newcommand{\Gj}{\smash{{\sG\kern-1.5pt}_q\kern-0.5pt(n,j)}}
\newcommand{\Grmk}{\smash{{\sG\kern-1.5pt}_q\kern-0.5pt(n,n-k)}}
\newcommand{\Grdk}{\smash{{\sG\kern-1.5pt}_q\kern-0.5pt(2k,k)}}
\newcommand{\Grekappa}{\smash{{\sG\kern-1.5pt}_q\kern-0.5pt(n,e+1-\kappa)}}
\newcommand{\Grtwoekappa}{\smash{{\sG\kern-1.5pt}_q\kern-0.5pt(n,2e+1-\kappa)}}
\newcommand{\Gremkappa}{\smash{{\sG\kern-1.5pt}_q\kern-0.5pt(n,e-\kappa)}}
\newcommand{\Gn}{\smash{{\sG\kern-1.5pt}_2\kern-0.5pt(n,n{-}1)}}
\newcommand{\Gnq}{\smash{{\sG\kern-1.5pt}_q\kern-0.5pt(n,n{-}1)}}
\newcommand{\Gone}{\smash{{\sG\kern-1.5pt}_2\kern-0.5pt(n,1)}}
\newcommand{\Gqone}{\smash{{\sG\kern-1.5pt}_q\kern-0.5pt(n,1)}}
\newcommand{\GTwo}{\smash{{\sG\kern-1.5pt}_2\kern-0.5pt(n,k)}}
\newcommand{\GTwonk}[2]{{\smash{{\sG\kern-1.5pt}_2\kern-0.5pt({#1},{#2})}}}
\newcommand{\Gnk}{\smash{{\sG\kern-1.5pt}_2\kern-0.5pt(n,n{-}k)}}
\newcommand{\Greone}{\smash{{\sG\kern-1.5pt}_q\kern-0.5pt(n,e{+}1)}}
\newcommand{\Gretwo}{\smash{{\sG\kern-1.5pt}_q\kern-0.5pt(n,e{+}2)}}

\newcommand{\be}[1]{\begin{equation}\label{#1}}
\newcommand{\ee}{\end{equation}}

\newcommand{\Cref}[1]{Co\-rol\-la\-ry\,\ref{#1}}

\newtheorem{theorem}{Theorem}
\newtheorem{lemma}{Lemma}

\newtheorem{corollary}{Corollary}

\newtheorem{definition}{Definition}

\newtheorem{example}{Example}

\newtheorem{problem}{Problem}

\newcommand{\e}[1]{\textcolor{red}{#1}}

\begin{document}

 \author{%
   \IEEEauthorblockN{\textbf{Shubhransh~Singhvi}\IEEEauthorrefmark{1}, 
                     \textbf{Omer~Sabary}\IEEEauthorrefmark{3}, 
                     \textbf{Daniella~Bar-Lev}\IEEEauthorrefmark{3}
                     and \textbf{Eitan~Yaakobi}\IEEEauthorrefmark{3}}
                     
   \IEEEauthorblockA{\IEEEauthorrefmark{1}%
                      Signal Processing  \&  Communications Research  Center, International Institute  of  Information Technology, Hyderabad, India}

                        \IEEEauthorblockA{\IEEEauthorrefmark{3}%
                     Department of Computer Science, %\\
                     Technion---Israel Institute of Technology, 
                     Haifa 3200003, Israel}
 }

\title{\textbf{The Input and Output Entropies of the $k$-Deletion/Insertion Channel} }
\date{\today}
\maketitle
\thispagestyle{empty}	
\pagestyle{empty}
%%%%%%%%

%%%%%%%%
\begin{abstract}
The channel output entropy of a transmitted  word is the entropy of the possible channel outputs and similarly the input entropy of a received word is the entropy of all possible transmitted words. The goal of this work is to study these entropy values for the $k$-deletion, $k$-insertion channel, where exactly $k$ symbols are deleted, inserted in the transmitted word, respectively. If all possible words are transmitted with the same probability then studying the input and output entropies is equivalent. For both the 1-insertion and 1-deletion channels, it is proved that among all words with a fixed number of runs, the input entropy is minimized for words with a skewed distribution of their run lengths and it is maximized for words with a balanced distribution of their run lengths. Among our results, we establish a conjecture by Atashpendar et al. which claims that for the binary 1-deletion, the input entropy is maximized for the alternating words. This conjecture is also verified for the $2$-deletion channel, where it is proved that constant words with a single run minimize the input entropy.  %A $q$-ary sequence of length $n$ is transmitted to the destination via a $k$-deletion channel, i.e., the channel deletes any $k$ symbols from the transmitted sequence. We investigate the $k$-deletion channel to analytically characterize the entropy of the posterior, i.e., entropy of the message source given the channel output. We prove that among all subsequences with a fixed number of runs (say $R$), entropy of the posterior is minimized and maximized by the subsequences with the most skewed distribution and the most uniform distribution of run lengths respectively. As a result, we prove that among all $q$-ary subsequences of length $n-1$, entropy of the posterior is minimized and maximized by the constant and the alternating subsequences respectively.
\end{abstract}

\section{Introduction}

%  \db{TODO: add definitions and problem statement}
%  \db{TODO: change claims to lemmas}
%\db{TBD - talk about the entropy and how it is related to the channel capacity.}
In the last decade, channels that introduce insertion and deletion errors attracted significant attention due to their relevance to DNA storage systems~\cite{Anavy433524, GHP15, OAC17, PTH21, TWB20, YGM16}, where deletions and insertions are among the most dominant errors~\cite{HMG18, SOSAYY19}. The study of communication channels with insertion and deletion errors is also relevant to many other applications such as %understanding 
the synchronization of files and symbols of data streams~\cite{DolecekAnan_Sync_2007} and for cases of over-sampling and under-sampling at the receiver side~\cite{SSB16}. \emph{VT codes}, designed by Varshamov and Tenengolts~\cite{VT}, are the first family of codes that correct a single deletion or a single insertion~\cite{L66}. Later, several works extended the scheme to correct multiple deletion errors, as well as substitutions; see e.g.,~\cite{BGZ16, GS17, SWGY17, SB19, SRB21}.  Under some applications, such as DNA storage systems, the problem of list decoding was studied as well, for example in~\cite{maria1, guruswami2020optimally, hayashi2018list, hanna2019list, liu2019list,  wachter2017list}. Under this setup, the decoder receives a channel output and returns a (short) list of possible codewords that includes the transmitted codeword. Even though many works considered channels with deletion errors and designed codes that correct such errors, finding optimal codes for these channels and their capacity are yet far from being solved~\cite{CK15, KD10, ED06, RD15}. %In this case, if additional information is given about the code, it can be applied on the returned list to detect the correct codeword. 

Additionally, noisy channels that introduce deletion and insertion errors were also studied as part of \emph{the trace reconstruction problem}~\cite{BKK04} and the \emph{sequence reconstruction problem}~\cite{L011, L012}. In these problems, a codeword $\bfx$ is transmitted over the same channel multiple times. This transmission results in several noisy copies of $\bfx$, and the goal is to find the required minimum number of these noisy copies that enables the reconstruction of $\bfx$ with high probability or in the worst case. Theoretical bounds and other results for the trace reconstruction problem over the deletion channel were proved in several works such as~\cite{AVDG19, BKK04, D21, gabrys2018sequence, HPP18}, and other works also studied algorithms for the sequence reconstruction problem for channels that introduce deletion and insertion errors; see e.g.,~\cite{SB21, CDR21, ISIT20, yaakobi2013sequence}.

%These algebraic codes were designed in 1961 by Varshamov and Tenengolts~\cite{} to correct symmetric substitution errors, but later was proven by Levenshtein~\cite{} to correct single deletion or single insertion errors. The VT-code original scheme was later extended to correct multiple errors in~\cite{}. 

This work studies several related problems to the capacity of the deletion and insertion channel and can contribute to their analysis. Assume a word $\bfx$ is transmitted over a communication channel and the channel output is $\bfy$. This work first considers the \emph{output entropy}, which is the entropy of all possible channel outputs $\bfy$ for a given input word $\bfx$. Similarly, we study the \emph{input entropy}, first studied in~\cite{ABC18}, which is the entropy of all possible transmitted words $\bfx$, given a channel output $\bfy$. Note that unlike the output entropy, the input entropy also depends on the probability to transmit each word $\bfx$.  Hence, it is assumed in this work that this probability is equal for each word $\bfx$. Our main goal is to characterize these entropies, their expected values, and the sequences that maximize and minimize them for a combinatorial version of the insertion and the deletion channels, which are referred as the \emph{$k$-deletion channel} and the \emph{$k$-insertion channel}. In the $k$-deletion, $k$-insertion channel the number of deletions, insertions is exactly $k$ and the errors are equally distributed~\cite{BGSY21, G15}. It was shown in~\cite{ABC18} that the sequences that minimize the input entropy for the $1$-deletion and the $2$-deletion channels, in the binary case, are the all-zeroes and all-ones words, under the equal transmittion probability assumption. %These channels were also studied in other works; see e.g.,.  
%  and  In the  channel, exactly $k$ symbols are selected randomly and deleted from the transmitted word. Similarly, in the $k$-insertion channel, exactly $k$ symbols are inserted to the trasmited word, in location that are selected randomly, with uniform probability to insert any symbol from the alphabet. 
%This paper studies the sequences that maximizes the output entropy, input entropy, a general expression for the input and output entropy and their expected values for the $k$-insertion and $k$-deletion channels, for the case where any word of given specific length is transmitted in equal probability.

Studying the input and output entropies for the insertion and deletion channels can assist with improving existing results in various applications. % where these errors are being considered. 
The output entropy is directly related to the channel capacity~\cite{CK15, D67,KD10, ED06, RD15}, and can be used in list-decoders to rank the different codewords from the decoder's list by their transmission probability~\cite{guruswami2020optimally,wachter2017list}. The output entropy can also be used to assist with code design by prioritizing words with lower output entropy since they are likely to have higher successful decoding probability. %that have higher  can be corrected with higher probability, and hence the minimum expected output entropy is a good metric to evaluate and construct codes. 
Alternatively, when the goal is to encrypt the information, words with higher output entropies are preferred~\cite{ARR15}. 
%\e{Our contribution can provide a better understanding of the behaviour of insertion and deletion errors, and assist with improving existing results in various applications where these errors are being considered. The output entropy is directly related to the channel capacity, and also can be used in list-decoders to rank the different codewords from the decoder's list. Information about the output entropy can also be used to assist with code design; when the goal is to correct errors, words with lower output entropy can be corrected with high probability, and hence the minimum expected output entropy is a good metric to evaluate codes. Alternatively, when the goal is to encrypt the information, words with higher output entropies are preffered.}    
Among other applications, the input entropy is relevant to DNA reconstruction algorithms, since in several of these algorithms such as~\cite{BO21, SYSY20, TrellisBMA}, there is a limitation on the number of noisy copies that are considered by the algorithm's decoder, due to design restrictions and run-time considerations. Therefore, in case the number of received noisy copies is greater than this limitation, a subset of these copies should be considered. Hence, to improve the accuracy of such algorithms, the input entropy of the channel outputs can be used to wisely select this subset of copies. Studying the input entropy was also investigated in~\cite{ABC18,ARR15} for analyzing the information leakage of a key by revealing any of its subsequences.

The rest of the paper is organized as follows. In Section~\ref{sec:pre}, we give the basic notations and definitions in the paper, discuss the deletion channel capacity, and describe its relation to the input and output entropies of the $k$-deletion and $k$-insertion channels. In this section we also give a formal definition of the problems studied in this paper. Section~\ref{sec:charact} characterizes the input entropies of the $k$-deletion and $k$-insertion channels and Section~\ref{sec:extremum} finds their extremum values for $k=1, 2$. Lastly, Section~\ref{sec:ave} studies the average input entropies for $k=1$.

\section{Preliminaries and Problem Statement}\label{sec:pre}
%\db{Add Capacity Definition, add the relation between in and out probs.}
Let $\cS$ be a discrete channel with input alphabet $\cX$ and output alphabet $\cY$. The channel is characterized by a conditional probability distribution 
$\mathsf{Pr}_\cS \{ \bfy \text{ received} | \bfx \text{ transmitted} \},$ for every $\bfx \in\cX^*$ and $\bfy\in\cY^*$, which is denoted for shorthand as $\mathsf{P}^{\mathsf{Out}}_\cS(\bfy | \bfx)$. 
%We use the short notation of $\mathsf{P}^{\mathsf{Out}}_\cS(\bfy | \bfx)$ to denote $ \mathsf{Pr}_\cS \{ \bfy \text{ received} | \bfx \text{ transmitted} \}$, and when the channel being discussed is clear, it will be omitted from the notation. 
The posterior probability of the channel $\cS$ is characterized by the conditional probability of the channel and the transmission probability of every message and is given by

$$ \mathsf{Pr}_\cS \{ \bfx \text{ transmitted} | \bfy \text{ received}\} =  \frac{ \mathsf{P}^{\mathsf{Out}}_\cS(\bfy | \bfx) \cdot \mathsf{P} \{ \bfx \text{ transmitted}  \} }{\mathsf{P} \{\bfy \text{ received}\} }.$$
 
\hspace{-0.5ex}We will refer to this probability in short as $\mathsf{P}^{\mathsf{In}}_\cS(\bfx | \bfy).$ When the channel will be clear from the context we may remove it from these notations. 
%It should be noted that, $$\mathsf{P}^{\mathsf{In}}_\cS(\bfx | \bfy) = \frac{ \mathsf{P}^{\mathsf{Out}}_\cS(\bfy | \bfx) \cdot \mathsf{P} \{ \bfx \text{ transmitted}  \} }{\mathsf{P} \{\bfy \text{ received}\} }
A channel $\cS$ is called a \emph{discrete memoryless channel} (\emph{DMC}) if it maps
$\cX^n$ to $\cY^n$ and

$$\mathsf{Pr}^{\mathsf{Out}}_\cS(y_1,\ldots,y_n|x_1,\ldots,x_n) = \prod_{i=1}^n \mathsf{Pr}^{\mathsf{Out}}_\cS(y_i|x_i).$$

It is well known that the channel capacity, i.e., the supremum of all achievable rates, of a DMC $\cS$ is given by

\begin{equation}\label{eq:cap}
\mathsf{Cap}(\cS) = \max_{\mathsf{P}(\bfx)} I(X;Y),  
\end{equation}
 
\hspace{-0.5ex}where $\mathsf{P}(\bfx)$ is the transmission probability for every $\bfx \in\cX^*$, $X,Y$ are random variables denoting the transmitted and the received symbols over the channel, respectively, and ${I(X;Y) = H(X) - H(X|Y) = H(Y) - H(Y|X)}$ is the mutual information between $X$ and $Y$. 

Synchronization channels, such as the insertion channel and the deletion channel, are not memoryless and thus the connection in (\ref{eq:cap}) does not necessarily hold for them. Hence, it is more common to study the capacity of the \emph{finite block length}, which is denoted by $\mathsf{Cap}_n(\cS)$ and is defined as

\begin{equation}\label{eq:cap_n}
\mathsf{Cap}_n(\cS) = \frac{1}{n}\max_{\mathsf{P}(\bfx)} I(X^{(n)};Y_{X^{(n)}}), 
\end{equation}

\hspace{-2.1ex}where the maximum is taken over all distributions $\mathsf{P}(\bfx)$ supported on $\cX^n$ and $Y_{X^{(n)}}$ is the random variable corresponding to the received message in this case. In general, it does not necessarily hold for synchronization channels that $\mathsf{Cap}(\cS) = \lim_{n\rightarrow\infty}\mathsf{Cap}_n(\cS)$, however, based upon Dobrushin's result~\cite{D67}, several works explored conditions for such an equality to hold. Note that

\begin{align*}
I(X^{(n)};Y_{X^{(n)}}) & = H(Y_{X^{(n)}}) - H(Y_{X^{(n)}}|X^{(n)}) & \\
&= H(Y_{X^{(n)}}) - \sum_{\bfx\in\cX^n}\mathsf{Pr}(\bfx)H(Y_{X^{(n)}}|X^{(n)} = \bfx), &
\end{align*}

\hspace{-2.5ex}or, alternatively,

\begin{align*} 
I(X^{(n)};Y_{X^{(n)}}) & = H(X^{(n)}) - H(X^{(n)} | Y_{X^{(n)}}) & \\
&= H(X^{(n)}) - \sum_{\bfy\in\cY^*}\mathsf{Pr}(\bfy)H(X^{(n)}|Y_{X^{(n)}} = \bfy).  &  
\end{align*}

\hspace{-2ex}Hence, one of the more important tasks in studying the capacity of synchronization channels is to determine for every $\bfx\in\cX^n$ the conditional entropy $H(Y_{X^{(n)}}|X^{(n)}=\bfx)$ and similarly for every $\bfy\in\cY^*$ the conditional entropy $H(X^{(n)}|Y_{X^{(n)}} = \bfy)$. Formally, for an input $\bfx\in \cX^n$, we refer to $H(Y_{X^{(n)}}|X^{(n)}=\bfx)$ as the \emph{output entropy} of the channel and for an output $\bfy\in \cY^*$, $H(X^{(n)}|Y_{X^{(n)}} = \bfy)$ is the \emph{input entropy} of the channel. Note that as opposed to several discrete symmetric memoryless channels, such as the binary symmetric channel and the binary erasure channel, the input and output entropies for synchronization channels, and in particular for the insertion and deletion channels, depend on the specific choice of $\bfx$ and $\bfy$. Furthermore, while the output entropy depends solely on the channel $\cS$, the input entropy depends both on $\cS$ and the channel input distribution $\mathsf{Pr}(X^{(n)})$. The case in which the channel input distribution is uniform is referred as \emph{uniform transmission}. 
%\ey{Add motivation for this problem and state the result for the BSC, BEC, and Z channel. Check if that should be normalized.}

This work studies the output and input entropies, while focusing primarily on the insertion channel and the deletion channel. Furthermore, to simplify the analysis of these problems, we consider the special case of the \emph{$k$-deletion channel}, \emph{$k$-insertion channel} which deletes, inserts exactly $k$ symbols of the transmitted word uniformly at random, respectively. %Next, we give a formal description of these channels by their corresponding conditional probability distributions.

For a positive integer $n$, let $[n]\triangleq\{1,\ldots,n\}$ and let $\Sigma_q \triangleq \{0,1,\ldots,q-1\}$ be an alphabet of size $q$. For an integer $n\ge0$, let $\Sigma_q^n$ be the set of all sequences (words) of length $n$ over the alphabet $\Sigma_q$. For an integer $k$, $0\le k\le n$, a sequence $\bfy\in\Sigma_q^{n-k}$ is a  \emph{$k$-subsequence} of $\bfx\in\Sigma_q^n$ if $\bfy$ can be obtained by deleting $k$ symbols from $\bfx$. Similarly, a sequence $\bfy\in\Sigma_q^{n+k}$ is a  \emph{$k$-supersequence} of $\bfx\in\Sigma_q^n$ if $\bfx$ is a $k$-subsequence of~$\bfy$. Let $\bfx$ and $\bfy$ be two sequences of length $n$ and $m$ respectively such that $m < n$. The \emph{embedding number} of $\bfy$ in $\bfx$, denoted by $\omega_{\bfy}(\bfx)$, is defined as the number of distinct occurrences of $\bfy$ as a subsequence of $\bfx$. More formally, the embedding number is the number of distinct index sets, $(i_1, i_2, \ldots, i_{m})$ , such that $1 \le i_1 < i_2 < \cdots < i_{m} \le n$ and $x_{i_1} = y_1, x_{i_2} = y_2, \dotsc, x_{i_{m}} = y_{m}$. For example, for $\bfx = \texttt{11220}$ and $\bfy = \texttt{120}$, it holds $\omega_{\bfy}(\bfx) = 4$. The \textit{$k$-insertion ball} centred at ${\bfx\in\Sigma_q^n}$, denoted by $I_t(\bfx)\subseteq \Sigma_q^{n+t}$, is the set of all $k$-supersequences of $\bfx$. Similarly, the {\textit{$k$-deletion ball}} centred at ${\bfx\in\Sigma_q^n}$, denoted by $D_t(\bfx)\subseteq \Sigma_q^{n-k}$, is the set of all $k$-subsequences of~$\bfx$.

In the \emph{$k$-deletion channel}, denoted by $k\textrm{-}\mathsf{Del}$, exactly $k$ symbols are deleted from the transmitted word. The $k$ symbols are selected uniformly at random out of the $\binom{n}{k}$ symbol positions, where $n$ is the length of the transmitted word. Hence, the conditional probability of the $k\textrm{-}\mathsf{Del}$ channel is,
$\mathsf{Pr}_{k\textrm{-}\mathsf{Del}}^{\mathsf{Out}}\{\bfy| \bfx\}=\frac{\omega_{\bfy}(\bfx)}{\binom{n}{k}},$
for all $\bfx\in\Sigma_q^n,\bfy\in\Sigma_q^{n-k}$. Similarly, for the \emph{$k$-insertion channel}, denoted by $k\textrm{-}\mathsf{Ins}$, exactly $k$ symbols are inserted to the transmitted word, while the locations and values of the $k$ symbols are selected uniformly at random out of the $\binom{n+k}{k}$ possible locations and $q^k$ possible options for the symbols. Therefore, the conditional probability of the $k$-insertion channel is,
${\mathsf{Pr}_{k\textrm{-}\mathsf{Ins}}^{\mathsf{Out}}\{\bfy| \bfx\}=\frac{\omega_{\bfx}(\bfy)}{\binom{n+k}{k}q^k},}$ % 
for all $\bfx\in\Sigma_q^n,\bfy\in\Sigma_q^{n+k}$.

The goal of this work is to study the following values for the $k$-deletion and $k$-insertion channels.
%In our work, we consider the following problems for $k$-deletion channel and $k$-insertion channel with $\cC =\Sigma_q^n$. Unless otherwise stated we will assume uniform probability of transmission over the code. 
\begin{problem}\label{prob1}
Find the following values for the channel $k\textrm{-}\mathsf{Del}$:
\begin{enumerate}
    \item For all $\bfx\in\Sigma_q^n$, find its output entropy over $k\textrm{-}\mathsf{Del}$, 
     
    $$\mathsf{H}^{\mathsf{Out}}_{k\textrm{-}\mathsf{Del}} (\bfx) \triangleq H(Y_{X^{(n)}}|X^{(n)}=\bfx).$$
     
    \item Find the minimum, maximum, and average output entropy of the channel $k\textrm{-}\mathsf{Del}$, 
    \begin{align*}
    \mathsf{min}^{\mathsf{Out}}_{k\textrm{-}\mathsf{Del}}(n) &\triangleq \min_{\bfx\in\Sigma_q^n}\left\{\mathsf{H}^{\mathsf{Out}}_{k\textrm{-}\mathsf{Del}} (\bfx)\right\}, \\
    \mathsf{max}^{\mathsf{Out}}_{k\textrm{-}\mathsf{Del}}(n) &\triangleq \max_{\bfx\in\Sigma_q^n}\left\{\mathsf{H}^{\mathsf{Out}}_{k\textrm{-}\mathsf{Del}} (\bfx)\right\}, \\
    \mathsf{avg}^{\mathsf{Out}}_{k\textrm{-}\mathsf{Del}}(n) &\triangleq \mathbb{E}_{\bfx\in\Sigma_q^n}\left\{\mathsf{H}^{\mathsf{Out}}_{k\textrm{-}\mathsf{Del}} (\bfx)\right\}.
    \end{align*}
     
    %$$\mathsf{min}}^{\mathsf{Out}}_{\mathsf{Del}}(k) \triangleq \min_{\bfx\in\Sigma_q^n}\{\mathsf{H}^{\mathsf{Out}}_{k\textrm{-}\mathsf{Del}} (\bfx)\},$$
    %$$\mathsf{max}}^{\mathsf{Out}}_{\mathsf{Del}}(k) \triangleq \max_{\bfx\in\Sigma_q^n}\{\mathsf{H}^{\mathsf{Out}}_{k\textrm{-}\mathsf{Del}} (\bfx)\}.$$
    %$$\mathsf{avg}}^{\mathsf{Out}}_{\mathsf{Del}}(k) \triangleq \mathbb{E}_{\bfx\in\Sigma_q^n}\{\mathsf{H}^{\mathsf{Out}}_{k\textrm{-}\mathsf{Del}} (\bfx)\},$$
    \item For all $\bfy\in\Sigma_q^{n-k}$, find its input entropy over the channel $k\textrm{-}\mathsf{Del}$ with input distribution $\mathsf{Pr}(X^{(n)})$, 
     
    $$\mathsf{H}^{\mathsf{In}}_{k\textrm{-}\mathsf{Del},\mathsf{Pr}(X^{(n)})} (\bfy) \triangleq H(X^{(n)}|Y_{X^{(n)}} = \bfy).$$
     
    \item Find the minimum, maximum, and average input entropy of the channel $k\textrm{-}\mathsf{Del}$ with input distribution $\mathsf{Pr}(X^{(n)})$,  
    \begin{align*}
    \mathsf{min}^{\mathsf{In}}_{k\textrm{-}\mathsf{Del}}\left(n,\mathsf{Pr}\left(X^{(n)}\right)\right) &\triangleq \min_{\bfy\in\Sigma_q^{n-k}}\left\{\mathsf{H}^{\mathsf{In}}_{k\textrm{-}\mathsf{Del},\mathsf{Pr}(X^{(n)})} (\bfy)\right\}, \\
    \mathsf{max}^{\mathsf{In}}_{k\textrm{-}\mathsf{Del}}\left(n,\mathsf{Pr}\left(X^{(n)}\right)\right) &\triangleq \max_{\bfy\in\Sigma_q^{n-k}}\left\{\mathsf{H}^{\mathsf{In}}_{k\textrm{-}\mathsf{Del},\mathsf{Pr}(X^{(n)})} (\bfy)\right\},\\
     \mathsf{avg}^{\mathsf{In}}_{k\textrm{-}\mathsf{Del}}\left(n,\mathsf{Pr}\left(X^{(n)}\right)\right) &\triangleq \mathbb{E}_{\bfy\in\Sigma_q^{n-k}}\left \{\mathsf{H}^{\mathsf{In}}_{k\textrm{-}\mathsf{Del},\mathsf{Pr}(X^{(n)})} (\bfy)\right \}. 
    \end{align*}
    %$$\mathsf{min}^{\mathsf{In}}_{\mathsf{Del}}(k,\mathsf{Pr}(X^{(n)})) \triangleq \min_{\bfy\in\Sigma_q^{n-k}}\{\mathsf{H}}^{\mathsf{In}}_{k\textrm{-}\mathsf{Del},\mathsf{Pr}(X^{(n)})} (\bfy)\},$$
    %$$\mathsf{max}^{\mathsf{In}}_{\mathsf{Del}}(k,\mathsf{Pr}(X^{(n)})) \triangleq \max_{\bfy\in\Sigma_q^{n-k}}\{\mathsf{H}}^{\mathsf{In}}_{k\textrm{-}\mathsf{Del},\mathsf{Pr}(X^{(n)})} (\bfy)\}.$$
    %$$\mathsf{avg}^{\mathsf{In}}_{\mathsf{Del}}(k,\mathsf{Pr}(X^{(n)})) \triangleq \mathbb{E}_{\bfy\in\Sigma_q^{n-k}}\{\mathsf{H}}^{\mathsf{In}}_{k\textrm{-}\mathsf{Del},\mathsf{Pr}(X^{(n)})} (\bfy)\}.$$ 
      
\end{enumerate}
The equivalent values and notations are defined similarly for the $k$-insertion channel, $k\textrm{-}\mathsf{Ins}$.
\end{problem}
Atashpendar et al.~\cite{ABC18} studied the minimum input entropy of the $k\textrm{-}\mathsf{Del}$ for the case where $k \in \{ 1, 2\}$. In their work, they presented an efficient algorithm to count the number of occurrences of a sequence $\bfy$ in any given supersequence $\bfx$. They also provided an algorithm that receives a sequence $\bfy$, computes the set of all of its supersequences of specific length, characterizes the distribution of their embedding numbers and clusters them by their Hamming weight. Lastly, they proved that the all-zero and the all-one words minimize the input entropy of the $1\textrm{-}\mathsf{Del}$ and $2\textrm{-}\mathsf{Del}$ channels, under uniform transmission.

In the following lemma we show the relation between the input entropy and the output entropy of the $k$-insertion and the $k$-deletion channels under uniform transmission. % the assumption that any word is transmitted with the same probability.
\begin{lemma} 
For uniform transmission over the $k$-insertion, {$k$-deletion} channel and for any channel output $\bfy$, it holds that, %and for a  $k$-insertion channel with $(n-k)$-length input, the output $\bfy_\mathsf{Ins} \in \Sigma_q^{n}$. When probability of transmission is uniform over the entire input space,

$$\mathsf{H}^{\mathsf{In}}_{k\textrm{-}\mathsf{Del},\mathsf{Pr}(X^{(n)})} (\bfy) = \mathsf{H}^{\mathsf{Out}}_{k\textrm{-}\mathsf{Ins}} (\bfy) \textrm{ and } \mathsf{H}^{\mathsf{In}}_{k\textrm{-}\mathsf{Ins},\mathsf{Pr}(X^{(n)})} (\bfy) = \mathsf{H}^{\mathsf{Out}}_{k\textrm{-}\mathsf{Del}} (\bfy). $$ 
 
%\hspace{-2.2ex}and 
%   $$\mathsf{H}^{\mathsf{In}}_{k\textrm{-}\mathsf{Ins},\mathsf{Pr}(X^{(n)})} (\bfy) = \mathsf{H}^{\mathsf{Out}}_{k\textrm{-}\mathsf{Del}} (\bfy).$$    
\end{lemma}

\begin{proof}
Since any word $\bfx$ of length $n$ can be transmitted with the same probability, it can be shown that any channel output $\bfy$ of length $n-k$, $n+k$ can be obtained with the same probability as an output of the channel $k\textrm{-}\mathsf{Del}$, $k\textrm{-}\mathsf{Ins}$, respectively. Thus, the results can be derived directly from the definitions.
\end{proof}
Based on the latter lemma, the rest of the paper studies the input entropy of the $k$-insertion and the $k$-deletion channels. Unless otherwise stated, we shall assume uniform transmission and thus, we drop the $\mathsf{Pr}(X^{(n)})$ term from the notation of the input entropy.

\section{Characterization of the Input Entropy}\label{sec:charact}
This section presents a complete characterization of the input entropy of the channels $k\textrm{-}\mathsf{Del}$, $k\textrm{-}\mathsf{Ins}$ and gives an explicit expression of these entropies for any channel output $\bfy$.

%Firstly, we characterize the conditional input entropy and then using the characterization, we solve for the $\mathsf{min}$ and $\mathsf{max}$ entropies for $k=1$.  $k=2$. 

The following definitions will be used in the rest of the paper. 
%A subsequence of consecutive $\bfy$ of a sequence 
For a sequence $\bfx$, a \emph{run} of $\bfx$ is a maximal subsequence of identical consecutive symbols within $\bfx$.
%if it is a consecutive sequence of the same character (i.e., $\bfy \in \alpha^{+}$ for an $\alpha \in \Sigma_{q}$). 
The number of runs in $\bfx$ is denoted by $\rho(\bfx)$.
We denote by $\Sigma_{q,R}^n$, the set of sequences $\bfx \in \Sigma_q^n$, such that $\rho(\bfx)=R$. It is well known that the number of such sequences is $|\Sigma_{q,R}^n| = {\binom{n-1}{R-1}} q (q-1)^{R-1}$.
%Let the set of sequences of length $n$ having $R$ runs be denoted by $\Sigma_{q,R}^n$. 
%It can be verified that  $|\Sigma_{q,R}^n| = {n-1 \choose R -1}.q.(q-1)^{(R - 1)}$. 

%\begin{definition}
For a sequence $\bfx \in \Sigma_{q,R}^n$, its \emph{run length profile}, denoted by $\cR\cL(\bfx)$, is the vector of the lengths of the runs in $\bfx$. That is, %lengths of the $R$ runs of $x$, i.e,
$
\cR\cL(\bfx) \triangleq  (r_1,r_2,\ldots, r_R),
$
where $r_{i}$, for $i\in[R]$ denotes the length of the $i$-th run of $\bfx$. 
%\end{definition}
% \db{Add an example}
%\begin{example}
For example, for $q=4$ and $\bfx = \texttt{311221110}$, we have that $\cR\cL(\bfx) = (1, 2, 2, 3, 1)$. %as $r_1 = 1, r_2 = 2, r_3=2 ,r_4 =3$ and  $r_5=1$.
%\end{example}
It is said that $\bfx \in \Sigma_{q,R}^n$, is \textit{skewed} if it consists of $R-1$ runs of length one, and a single run of length $n-(R-1)$; $\bfx$ is \textit{balanced} if $\bfx$ consists of $r\equiv n \bmod R$ runs of length $\ceil{\frac{n}{R}}$ and the remaining $R - r$ runs are of length $\floor{\frac{n}{R}}$.
It is known that the size of the $k$-insertion ball of any word $\bfy\in \Sigma_q^{n-k}$ is $|I_k(\bfy)| = \sum_{0\le i \le k} \binom{n}{i}(q-1)^i$~\cite{L66}. 
%The proof for the following Lemma is omitted as it simply follows from the definitions. 
\begin{lemma} \label{lm:entropyIn-k-del}
For any integers $n$ and $k$, such that $k \le n$, and for any word $\bfy \in \Sigma_q^{n-k}$, we have that,

$$
\mathsf{H}^{\mathsf{In}}_{k\textrm{-}\mathsf{Del}} (\bfy) 
 = 
\log \left( \binom{n}{k} q^k\right) -\frac{1}{\binom{n}{k}q^k}\sum_{\bfx \in I_k(\bfy) } \omega_{\bfy}(\bfx) \cdot \log \left( {\omega_{\bfy}(\bfx)}\right).
$$
 
\end{lemma}
\begin{proof}
 
\begin{align*}
\mathsf{H}^{\mathsf{In}}_{k\textrm{-}\mathsf{Del}} (\bfy) &=
- \sum_{\bfx \in I_k(\bfy) }\left( \mathsf{Pr}_{k\textrm{-}\mathsf{Del}}^{\mathsf{In}}\{\bfx| \bfy\} \cdot \log \left( \mathsf{Pr}_{k\textrm{-}\mathsf{Del}}^{\mathsf{In}}\{\bfx| \bfy\} \right) \right) 
\\&=
- \sum_{\bfx \in I_k(\bfy) } \frac{\omega_{\bfy}(\bfx)}{\binom{n}{k} q^k } \cdot \log \left( \frac{\omega_{\bfy}(\bfx)}{\binom{n}{k} q^k }\right)  
\\ & = 
\log \left( \binom{n}{k} q^k\right)  -\frac{1}{\binom{n}{k}q^k}\sum_{\bfx \in I_k(\bfy) } \omega_{\bfy}(\bfx) \cdot \log \left( {\omega_{\bfy}(\bfx)}\right).
\end{align*} 
  
\end{proof}
In the special case of $k=1$, we have the next corollary.

%Let $m\triangleq n-1$. From \ref{}, it is known that the size of $1$-insertion ball is $|I_1(\bfy)| = q + |\bfy|(q-1)$.
 
\begin{corollary}
\label{cor:entropy-single-deletion}
$\mathsf{H}^{\mathsf{In}}_{1\textrm{-}\mathsf{Del}} (\bfy)$ is invariant to permutations of run length profile of the channel output $\bfy \in \Sigma_q^{n-1}$ and it is given by the following expression

\begin{align*}
\mathsf{H}^{\mathsf{In}}_{1\textrm{-}\mathsf{Del}} (\bfy) = \log(nq) - \frac{1}{nq}\sum_{i=1}^{\rho(\bfy)}(r_{i}+1)\log(r_{i}+1),
\end{align*}
 
where $r_{i}, i \in [\rho(\bfy)]$, denotes the length of the $i$-{th} run of $\bfy$. 
\end{corollary}

\begin{proof}
From~\cite{BGSY21, SBGYY22}, % we have that,  the  $1$-insertion ball,
$I_1(\bfy)$ has $\rho(\bfy)$ sequences that can be obtained by prolonging an existing run, each with embedding weight $r_{i}+1$ for $1 \le i \le \rho(\bfy)$.  The embedding weight of the remaining $|I_{1}(\bfy)| - \rho(\bfy) =  q+|\bfy|(q-1) - \rho(\bfy)$ sequences is one. Therefore from Lemma~\ref{lm:entropyIn-k-del}, it follows that,

\begin{align*}
\mathsf{H}^{\mathsf{In}}_{1\textrm{-}\mathsf{Del}} (\bfy) & = \log(nq) - \frac{1}{nq}\sum_{i=1}^{\rho(\bfy)}(r_{i}+1)\log(r_{i}+1).
\end{align*}

\begin{comment}
\begin{align*}
\begin{split}
\mathsf{H}^{\mathsf{In}}_{1\textrm{-}\mathsf{Del}} (\bfy) &= \left(\sum_{i=1}^{\rho(\bfy)}\frac{(r_{i}+1)}{nq}\log\left(\frac{nq}{(r_{i}+1)}\right)\right)\\
&+ (q+|\bfy|(q-1) - \rho(\bfy))\left(\frac{1}{nq}\log\left(nq\right)\right) 
\end{split}
\end{align*}
\begin{align*}
\begin{split}
&= \left(\sum_{i=1}^{\rho(\bfy)}\frac{(r_{i}+1)}{nq}\log\left(nq\right)\right) - \left(\sum_{i=1}^{\rho(\bfy)}\frac{(r_{i}+1)}{nq}\log\left(r_{i}+1\right)\right)\\
&+ (q+(n-1)(q-1) - \rho(\bfy))\left(\frac{1}{nq}\log\left(nq\right)\right) 
\end{split}
\\
\begin{split}
&= \left(\frac{(n-1+\rho(\bfy))}{nq}\log\left(nq\right)\right) - \left(\sum_{i=1}^{\rho(\bfy)}\frac{(r_{i}+1)}{nq}\log\left(r_{i}+1\right)\right)\\
&+ \frac{(nq - (n-1) - \rho(\bfy))}{nq}\log\left(nq\right)
\end{split}
\\
&= \log(nq) - \frac{1}{nq}\left(\sum_{i=1}^{\rho(\bfy)}(r_{i}+1)\log(r_{i}+1)\right)
\end{align*} 
\end{comment}
%Thus, as the entropy expression does not depend on the order of the runs, it is invariant to permutations of $\cR\cL(\bfy)$.  
\end{proof}
%  \db{TODO - this is a new result, right? if so add a proof for it.}
%  \db{use channel output instead of subsequences}
% \db{Problem description}

%\subsection{One-Insertion Channel}
Similarly, we have the next lemma for the $k\textrm{-}\mathsf{Ins}$ channel. 
%The next lemma presents this results for the $k$. 
%The next lemma summarizes the same results of Lemma~\ref{lm:entropyIn-k-del} and Corollary~\ref{cor:entropy-single-deletion} for the $k$-insertion channel. %Since the proof of this lemma is similar, it is omitted from the paper.  %The same results from Le
%Let $m\triangleq n+1$. 
%It is  well known that the size of $1$-deletion ball is $|D_1(\bfy)| = \rho(\bfy)$~\cite{L66}.
\begin{lemma}
\label{lm:entropyIn-k-ins}
For any integers $n$ and $k$, such that $k \le n$, and for any word $\bfy \in \Sigma_q^{n+k}$, we have that, 
$$\mathsf{H}^{\mathsf{In}}_{k\textrm{-}\mathsf{Ins}} (\bfy) = \log  \binom{n+k}{k} -\frac{1}{\binom{n+k}{k}}\sum_{\bfx \in D_k(\bfy) } \omega_{\bfy}(\bfx) \cdot \log \left( {\omega_{\bfy}(\bfx)}\right),$$
where $ D_k(\bfy)$ denotes the $k$-deletion ball~\cite{L66}. In particular, for $k=1$,   %$\mathsf{H}^{\mathsf{In}}_{1\textrm{-}\mathsf{Ins}} (\bfy)$ is invariant to permutations of run length profile of the channel output $\bfy$ and can be analytically computed by the following expression
$$\mathsf{H}^{\mathsf{In}}_{1\textrm{-}\mathsf{Ins}} (\bfy) = \log(n+1) -\frac{1}{n+1}\sum_{i=1}^{\rho(\bfy)}r_{i}\log(r_{i}),$$
where $r_{i}, i \in [\rho(\bfy)]$ denotes the length of the $i$-{th} run of $\bfy$.
\end{lemma}
\begin{comment}
\begin{proof}
From \ref{}, we know that the single deletion ball, $D_1(\bfy)$, has $\rho(\bfy)$ sequences with embedding weights, $r_{1}, r_{2}, \ldots, r_{\rho(\bfy)}$, obtained by shortening a run. Therefore, the total embedding weight for the single deletion ball is
\bean
\sum_{\bfx \in D_{1}(\bfy)}\omega_\bfx(\bfy) = \sum_{i=1}^{\rho(\bfy)} r_{i} = m. 
\eean
Therefore,
\begin{align*}
\mathsf{H}^{\mathsf{In}}_{k\textrm{-}\mathsf{Ins}} (\bfy) &= \left(\sum_{i=1}^{\rho(\bfy)}\frac{r_{i}}{m}\log\left(\frac{m}{r_{i}}\right)\right)\\
&= \left(\sum_{i=1}^{\rho(\bfy)}\frac{r_{i}}{m}\log\left({m}\right)\right) - \left(\sum_{i=1}^{\rho(\bfy)}\frac{r_{i}}{m}\log\left(\cR_i\right)\right)\\
&= \log(m) -\frac{1}{m}\left(\sum_{i=1}^{\rho(\bfy)}(r_{i})\log(r_{i})\right)
\end{align*}
Therefore, as the entropy expression does not depend on the order of the runs, it is invariant to permutations of $\cR\cL(\bfy)$.  
\end{proof}
\end{comment}

\section{Extremum Values of the Input Entropy}\label{sec:extremum}
In this section we first find the channel outputs that have maximum and minimum entropy among all sequences in $\Sigma_{q,R}^m$, where $R\in[m]$, for $m \in \{ n-1, n+1 \}$. Then as corollaries, we find the maximum and minimum values of the input entropies of the channels $1\textrm{-}\mathsf{Del}$ and $1\textrm{-}\mathsf{Ins}$. We then derive the minimum value of the input entropy for $2\textrm{-}\mathsf{Del}$ channel.

\subsection{The Single-Deletion Channel: Maximum Input Entropy}
This subsection studies the channel outputs $\bfy \in \Sigma_q^m$ that maximize the input entropy of the $1\textrm{-}\mathsf{Del}$ channel, where $m\triangleq n-1$. 

%In the next lemma it is proved that among all channel outputs with $R$ runs, the balanced ones maximize the input entropy.
\begin{lemma}\label{lm:min-entropy-R-runs}
For $R\in [m]$ and $r\equiv m \bmod R$, the maximum input entropy among all channel outputs in $\Sigma_{q,R}^{m}$ is
 \begin{align*}
\begin{split}
\max_{\bfy \in \Sigma_{q,R}^{m}}{\mathsf{H}^{\mathsf{In}}_{1\textrm{-}\mathsf{Del}} (\bfy)} &= \log(nq) -\frac{r}{nq}\left(\ceil{\frac{m}{R}}+1\right)\log\left(\ceil{\frac{m}{R}}+1\right) \\
&-\frac{R-r}{nq}\left(\floor{\frac{m}{R}}+1\right)\log\left(\floor{\frac{m}{R}}+1\right),
\end{split}
\end{align*}
 and it is attained only by balanced channel outputs. 
\end{lemma}
\begin{proof} 
Let $\bfy \in \Sigma_{q,R}^{m}$ be a sequence with maximum entropy and assume its run length profile is $(r_1,r_2,\ldots,r_R)$. Assume to the contrary that $\bfy$ is not one of the balanced sequences. Then, there exist indices $\ell\neq s$ such that $ r_s< \floor{\frac{m}{\rho(\bfy)}} $ and $\ceil{\frac{m}{\rho(\bfy)}} < r_\ell$ and in particular $r_\ell - r_s \geq2$. Consider a sequence $\bfy' \in \Sigma_{q,R}^{m}$, with run length profile $(r_1',r_2',\ldots,r_R')$, where $r'_\ell = r_\ell-1$, $r'_s=r_s+1$, and for any $k \notin \{ \ell, s \}$, $r'_k=r_k$.
\begin{comment}
\bean
r_k' = \begin{cases} 
      r_k-1 & k = l \\
      r_k+1 & k = s\\
      r_k & \text{otherwise.} 
 \end{cases}
 \eean
\end{comment}

Next, consider the entropies difference  
\begin{align*}
\mathsf{H}^{\mathsf{In}}_{1\textrm{-}\mathsf{Del}} (\bfy') \hspace{-.5ex}-\hspace{-.5ex} \mathsf{H}^{\mathsf{In}}_{1\textrm{-}\mathsf{Del}} (\bfy)
%\begin{split}
&\hspace{-.5ex}=-\frac{1}{nq}\big(r_{\ell}\log r_{\ell}+ (r_{s}+2)\log(r_{s}+2) & \\
& -(r_{\ell}\hspace{-.5ex}+\hspace{-.5ex}1)\log(r_{\ell}\hspace{-.5ex}+\hspace{-.5ex}1)\hspace{-.5ex}-\hspace{-.5ex} (r_{s}\hspace{-.5ex}+\hspace{-.5ex}1)\log(r_{s}\hspace{-.5ex}+\hspace{-.5ex}1) \big).& 
%((r_{\ell}\hspace{-.5ex}+\hspace{-.5ex}1)\log(r_{\ell}\hspace{-.5ex}+\hspace{-.5ex}1) \hspace{-.5ex}+\hspace{-.5ex} (r_{s}\hspace{-.5ex}+\hspace{-.5ex}1)\log(r_{s}\hspace{-.5ex}+\hspace{-.5ex}1) )\\ & \ \ \ -\frac{1}{nq}\left((r_{\ell})\log(r_{\ell}) + (r_{s}+2)\log(r_{s}+2)\right).
\end{align*}
 
\hspace{-2.1ex}Let $g(r) \triangleq (r+1)\log(r+1) - r\log(r)$, and note that $g$ is an increasing function w.r.t. $r$. In addition, since $r_\ell > r_s+1$,
\begin{align*}
    &\mathsf{H}^{\mathsf{In}}_{1\textrm{-}\mathsf{Del}} (\bfy') - \mathsf{H}^{\mathsf{In}}_{1\textrm{-}\mathsf{Del}} (\bfy) = \frac{g(r_\ell) - g(r_s+1)}{nq} > 0.
\end{align*}
 
This is a contradiction to the assumption on $\bfy$. Therefore, among all sequences in $\Sigma_{q,R}^{m}$, the maximum entropy is attained by the balanced channel outputs. The value of this maximum entropy can be simply derived from the run length profile of the balanced channel outputs. 
% 
% 
%\begin{align*}
%\begin{split}
%\max_{\bfy \in\Sigma_{q,R}^{m}}{\mathsf{H}^{\mathsf{In}}_{1\textrm{-}\mathsf{Del}} (\bfy)} %&= \log(nq) -\frac{r}{nq}\left(\ceil{\frac{m}{R}}+1\right)\log\left(\ceil{\frac{m}{R}}+1\right) \\
%&-\frac{R-r}{nq}\left(\floor{\frac{m}{R}}+1\right)\log\left(\floor{\frac{m}{R}+1\right).
%\end{split} 
%\end{align*} 
% 
\end{proof}

%The next corollary states that the input entropy is maximized by channel outputs with exactly $m$ runs (each of length one). 
\begin{corollary} 
The maximum input entropy among all channel outputs in $\Sigma_{q}^{m}$ is
$$\mathsf{max}^{\mathsf{In}}_{1\textrm{-}\mathsf{Del}}(n)=\log(nq) -\frac{2m}{nq}$$
and is only attained by channel outputs with $m$ runs.
\end{corollary} 
\begin{proof} 
Let $\bfy \in \Sigma_{q}^{m}$ and let $\cR\cL(\bfy) = (r_1,r_2,\ldots,r_{\rho(\bfy)})$. From Corollary \ref{cor:entropy-single-deletion}, it can be shown that,  
\begin{align*}
\mathsf{H}^{\mathsf{In}}_{1\textrm{-}\mathsf{Del}} (\bfy) 
& = \log(nq)+ \frac{m+\rho(\bfy)}{nq}\sum_{i=1}^{\rho(\bfy)}\frac{r_{i}+1}{m+\rho(\bfy)}\log\left( \frac{m+\rho(\bfy)}{r_{i}+1}\right) \\  
& \ \ \ -\frac{m+\rho(\bfy)}{nq}\log(m+\rho(\bfy)).% 
\end{align*}% 
 
Since $\displaystyle\sum_{i=1}^{\rho(\bfy)}\frac{r_{i}+1}{m+\rho(\bfy)} = 1$, Jensen's inequality implies that
$$
\mathsf{H}^{\mathsf{In}}_{1\textrm{-}\mathsf{Del}} (\bfy) \leq \log(nq) + \frac{m+\rho(\bfy)}{nq}\log\left(\frac{\rho(\bfy)}{m+\rho(\bfy)}\right). 
$$
 
Let $f(x):[1,m]\to \R$ be defined as 
$$
f(x) \triangleq \log(nq) + \frac{m+x}{nq}\log\left(\frac{x}{m+x}\right). 
$$
$f$ is increasing w.r.t $x$ and
\begin{comment}
\begin{align*}
    \frac{d}{dx} f(x)=\frac{1}{nq\ln(2)}\left(\frac{m}{x}-\log\left(\frac{m+x}{x}\right)\right) \geq 0.
\end{align*}
\end{comment}
hence, for $r \in [1,m]$, $f(r) \leq f(m)$, and equality is attained if and only if $r=m$. Hence, among all sequences in $\Sigma_q^{m}$, channel outputs with $\rho(\bfy) = m$ have the maximum entropy $\mathsf{max}^{\mathsf{In}}_{1\textrm{-}\mathsf{Del}}(n) = \log(nq) -\frac{2m}{nq}.$
%Hence, when $\rho(\bfy) = m$, 
%\bean
%\mathsf{H}^{\mathsf{In}}_{k\textrm{-}\mathsf{Ins}} (\bfy) = \log(nq) -\frac{2m}{nq}.
%\eean
%Therefore, among all sequences in $\Sigma_q^{m}$, channel outputs with $\rho(\bfy) = m$ have maximum entropy. \ey{we need the strictly increasing property for that.}
\end{proof}

\subsection{The Single-Deletion Channel: Minimum Input Entropy}
Similarly to the previous subsection, here we study the channel outputs $\bfy \in \Sigma_q^m$ that minimize the input entropy. %The next lemma characterizes that minimum input entropy for channel outputs with $R$ runs, and proves that skewed channel outputs minimize it.
%The proof of the subsequent lemma follows the same ideas as in Lemma~\ref{lm:min-entropy-R-runs} and therefore is omitted from the paper.

\begin{lemma}\label{lem:min_entropy}
 Let $R\in [m]$, the minimum input entropy among all channel outputs in  $\Sigma_{q,R}^{m}$ is
\bean
\min_{\bfy \in \Sigma_{q,R}^{m}} {\mathsf{H}^{\mathsf{In}}_{1\textrm{-}\mathsf{Del}} (\bfy)} &=\log(nq) -\frac{(m-R+2)\log(m-R+2) + 2(R-1)}{nq}
\eean
and it is attained only by skewed channel outputs. 
\end{lemma}
\begin{proof} 
Let $\bfy \in \Sigma_{q,R}^{m}$ be the sequence with minimum entropy. Assume to the contrary that $\bfy$ is not one of the skewed sequences. Then there exist indices $\ell\neq s$ such that $1< r_s \leq r_\ell < m-(R-1)$. Consider the sequence $\bfy' \in \Sigma_{q,R}^{m}$, with run length profile $(r_1',r_2',\ldots,r_R')$, where
\bean
r_k' = \begin{cases} 
      r_k+1 & k = \ell \\
      r_k-1 & k = s\\
      r_k & \text{otherwise.} 
 \end{cases}
 \eean
Now, consider the entropies difference
\begin{align*}
% \begin{split}
% \mathsf{H}^{\mathsf{In}}_{1\textrm{-}\mathsf{Del}} (\bfy') - \mathsf{H}^{\mathsf{In}}_{1\textrm{-}\mathsf{Del}} (\bfy) &= \log(nq) - \frac{1}{nq}\left(\sum_{i=1}^{\rho(\bfy)}(r_{i}'+1)\log(r_{i}'+1)\right) \\
% &- \left(\log(nq) - \frac{1}{nq}\left(\sum_{i=1}^{\rho(\bfy)}(r_{i}+1)\log(r_{i}+1)\right)\right)
% \end{split}
% \\
&\mathsf{H}^{\mathsf{In}}_{1\textrm{-}\mathsf{Del}} (\bfy') - \mathsf{H}^{\mathsf{In}}_{1\textrm{-}\mathsf{Del}} (\bfy)\\
\begin{split}
&=\frac{1}{nq}\left((r_{\ell}+1)\log(r_{\ell}+1) + (r_{s}+1)\log(r_{s}+1) \right)\\
&-\frac{1}{nq}\left((r_{\ell}+2)\log(r_{\ell}+2) + (r_{s})\log(r_{s})\right)
\end{split}
\\
&=\frac{h(r_s) - h(r_\ell+1)}{nq},
\end{align*}
where $h$ is an increasing function defined in the proof of the previous lemma.  Therefore, as $r_s < r_\ell+1$ 
\begin{align*}
    &\mathsf{H}^{\mathsf{In}}_{1\textrm{-}\mathsf{Del}} (\bfy') - \mathsf{H}^{\mathsf{In}}_{1\textrm{-}\mathsf{Del}} (\bfy) = \frac{h(r_s) - h(r_\ell+1)}{nq} < 0.
\end{align*}
This is a contradiction as $\bfy$ has minimum entropy among all sequences in  $\Sigma_{q,R}^{m}$. Thus,  $\bfy$ is skewed and its input entropy is
%sequence is wrong and the skewed sequences have the minimum entropy among all sequences in $\Sigma_{q,R}^{m}$. 
%. Therefore, among all sequences in $\Sigma_{q,R}^{m}$ skewed sequences have minimum entropy and is equal to 
\bean
\min_{\bfy \in \Sigma_{q,R}^{m}}{\mathsf{H}^{\mathsf{In}}_{1\textrm{-}\mathsf{Del}} (\bfy)} &=\log(nq) -\frac{(m-R+2)\log(m-R+2) + 2(R-1)}{nq}.
\eean
\end{proof}
 
The next corollary states that the channel outputs that minimize this entropy have a single run. This extends the results from~\cite{ABC18} to the non-binary case. %since we present them in this paper with alternative proofs and give them in this paper for completeness.
\begin{corollary}\label{cor:min1del}
The minimum input entropy among all channel outputs in $\Sigma_{q}^{m}$ is
$
\mathsf{min}^{\mathsf{In}}_{1\textrm{-}\mathsf{Del}}(n)=\log(nq) -\frac{\log(n)}{q}
$
and is only attained by channel outputs having a single run.
\end{corollary}
\begin{proof}
According to Lemma~\ref{lem:min_entropy}, we have that
$$\min_{ \bfy \in \Sigma_{q,R}^{m}}{ \hspace{-1.5ex} \mathsf{H}^{\mathsf{In}}_{1\textrm{-}\mathsf{Del}} (\bfy)} \hspace{-.5ex}=\hspace{-.5ex}\log(nq) -\frac{(m\hspace{-.4ex}-\hspace{-.4ex}R\hspace{-.4ex}+\hspace{-.4ex}2)\log(m\hspace{-.4ex}-\hspace{-.4ex}R\hspace{-.4ex}+\hspace{-.4ex}2) \hspace{-.4ex}+ \hspace{-.4ex}2(R\hspace{-.4ex}-\hspace{-.4ex}1)}{nq}.$$
It is easy to verify that the input entropy decreases as $R$ decreases. Therefore,
$\mathsf{min}^{\mathsf{In}}_{1\textrm{-}\mathsf{Del}}(n) = \log(nq) -\frac{\log(n)}{q} $ and it is attained if and only if $\rho(\bfy) = 1$.
\end{proof}

\subsection{The Single-Insertion Channel}
Using similar techniques as in the previous subsections, we can analyze the $1\textrm{-}\mathsf{Ins}$ channel for  $m\triangleq n+1$. 
%The previous subsections characterized channel outputs that maxmimize and minimize the input entropy of the $1$-deletion channel. Similar technique can be used to prove equivalent results for the $1$-insertion channel. %Due to this similarity and the lack of space, the proofs are omitted. % of the following lemmas and corollaries from the paper.
 %Let $m\triangleq n+1$. 
\begin{theorem} 
If $R\in [m]$ and $r\equiv m \bmod R$, then % maximum input entropy among all channel outputs in $\Sigma_{q,R}^{m}$ is
\bean
\max_{\bfy \in \Sigma_{q,R}^{m}} \hspace{-1.4ex}{\mathsf{H}^{\mathsf{In}}_{1\textrm{-}\mathsf{Ins}} (\bfy)}\hspace{-.7ex}=\hspace{-0.5ex}\log(m)\hspace{-0.65ex} -\hspace{-0.65ex} \frac{1}{m}\left(r\hspace{-0.5ex}\ceil{\frac{m}{R}}\log\ceil{\frac{m}{R}} \hspace{-.3ex}+\hspace{-.3ex} (R\hspace{-.3ex}-\hspace{-.3ex}r)\floor{\frac{m}{R}}\log\floor{\frac{m}{R}}\right),
\eean
and the maximum is obtained only by balanced channel outputs. Furthermore, $\mathsf{max}^{\mathsf{In}}_{1\textrm{-}\mathsf{Ins}}(n)=\log(m)$ and it is attained only by channel outputs with $m$ runs.
\end{theorem}
\begin{comment}
\begin{corollary} The maximum input entropy among all subsequences in $\Sigma_{q}^{m}$ is
$
\mathsf{max}^{\mathsf{In}}_{\mathsf{Ins}}(1)=\log(n) 
$
and the maximum value is attained by channel outputs with $\rho(\bfy) = m$.
\end{corollary}
\end{comment}
\begin{theorem} If $R\in [m]$, then % minimum input entropy among all channel outputs in $\Sigma_{q,R}^{m}$ is
\begin{align*}
\min_{\bfy \in \Sigma_{q,R}^{m}}{\mathsf{H}^{\mathsf{In}}_{1\textrm{-}\mathsf{Ins}} (\bfy)}=\log(m) -\frac{(m-R+1)\log(m-R+1)}{m}
\end{align*}
and the minimum is obtained only by skewed channel outputs. Furthermore, 
$\mathsf{min}^{\mathsf{In}}_{1\textrm{-}\mathsf{Ins}}(n)= 0$
and is attained only by channel outputs with a single run. 
\end{theorem}

\subsection{The Double-Deletion Channel: Minimum Input Entropy}
This section studies the channel outputs $\bfy \in \Sigma_2^m$ that minimize the input entropy of the $2\textrm{-}\mathsf{Del}$ channel. Let $\cS \subseteq  I_k(\bfy)$ and define $\mathsf{W}_{\cS}(\bfy) \triangleq \sum_{\bfx \in \cS } \omega_{\bfy}(\bfx) \cdot \log \left( {\omega_{\bfy}(\bfx)}\right)$. From Lemma \ref{lm:entropyIn-k-del}, for $n=k+m$ we know that 
$$\mathsf{H}^{\mathsf{In}}_{k\textrm{-}\mathsf{Del}} (\bfy) = \log \left( \binom{n}{k} 2^k\right) -\frac{1}{\binom{n}{k}2^k}\sum_{\bfx \in I_k(\bfy) } \omega_{\bfy}(\bfx) \cdot \log \left( {\omega_{\bfy}(\bfx)}\right).$$
Therefore, $\underset{\bfy \in \Sigma_2^{m}}{\argmin}~\mathsf{H}^{\mathsf{In}}_{k\textrm{-}\mathsf{Del}} (\bfy) = \underset{\bfy \in \Sigma_2^{m}}{\argmax}~\mathsf{W}_{I_k(\bfy)}(\bfy)$.

Let $\bfy\in\Sigma_{2,R}^{m}$ with $\cR\cL(\bfy)=(r_1,r_2,\ldots,r_R)$. For $i\in[0,R-1]$, let $f_i$ denote the smallest index in $[i+1,R]$ such that $r_{f_i}>1$, and let $f_R=R$. If such an index does not exist then let $f_i=R$. Similarly, for $i \in [1,R]$, let $b_i$ denote the largest index in $[1,i]$ such that $r_{b_i}>1$, and let $b_0=1$. If such an index does not exist then let $b_i=1.$ 
 
\begin{lemma}\label{correction}
Let $\bfy\in\Sigma_{2,R}^{m}$, such that $\cR\cL(\bfy)=(r_1,\ldots,r_R)$ and let $\bfx \in I_2(\bfy)$. If $\cR\cL(\bfx)=(r_1,\ldots,r_i,1,1,r_{i+1},\ldots r_R)$, then
${\omega_{\bfy}(\bfx)} = 1+ \sum_{j = b_i}^{f_i} r_j.$
\end{lemma}
\begin{proof}
In \cite{ABC18}, ${\omega_{\bfy}(\bfx)}$ was characterized as $r_{i}+r_{i+1}+1$. In the case where $r_i=1$ or  $r_{i+1}=1$ the analysis is a bit harder and the characterization is slightly different as follows. Let $\bfz$ denote the subsequence of $\bfy$ from the $b_i$-th run to the $f_i$-th run. Note that $\bfz$ is an alternating segment of length $\sum_{j = b_i}^{f_i} r_j$. Since $\bfx$ is obtained from $\bfy$ by adding two runs of length one after the $i$-th run, we can consider the subsequence $\bfz'$ of $\bfx$ which is the sequence $\bfz$ with these additional two runs. 
%Let this elongated alteranting segment in $\bfx$ be denoted by $\bfz'$.
It can be verified that deleting any two bits in $\bfx$ that do not belong to $\bfz'$  will result with a sequence which is different than $\bfy$. Therefore, $\omega_{\bfy}(\bfx) = \omega_{\bfz}(\bfz').$
Observe that $\omega_{\bfz}(\bfz') = |\bfz| + 1$ and hence, $\omega_{\bfy}(\bfx) = 1+ \sum_{j = b}^{f} r_j $.
\end{proof}

Let $\alpha \in \Sigma_2, \bfx \in \Sigma_2^n$, and denote by $\alpha \circ \bfx$ their concatenation. 
%Let $\circ$ denote the concatenation operator such that for $\alpha \in \Sigma_2, \bfx \in \Sigma_2^n$, let $\alpha \circ \bfx$ be the  \triangleq \alpha x_1\ldots x_n \in \Sigma_2^{n+1}$. 
By abuse of notation, given a set of sequences $\cS$, we let $\alpha \circ \cS \triangleq \{ \alpha\circ \bfx : \bfx \in \cS \}$. Next, we show how $\omega_\bfy(\bfx)$ is affected when a bit $\alpha$ is appended at the beginning of both $\bfx$ and $\bfy$.% are  the effect on the embedding number when both sequences are appended with the same bit at the first index.
 
\begin{lemma}\label{alpha}
Let $\bfy \in \Sigma_2^{m}$ and $\bfx \in \Sigma_2^{n}$ with $m=n-k$. For $\bfy' = \alpha\circ\bfy$ and $\bfx' = \alpha \circ \bfx$ it holds that   
$$\omega_{\bfy'}(\bfx') = \omega_{\bfy}(\bfx) + \sum_{i=1}^{k}\omega_{\bfy}(\bfx_{[i+1,n]})\cdot \mathbb{I}_{\alpha = x_i},$$
where $\mathbb{I}$ is the indicator function. 
\end{lemma}
\begin{proof}
It can be verified that
\begin{align}\label{DP-EW}
 {\omega_{\bfy}(\bfx)} &= \begin{cases} 
     \omega_{\bfy_{[2,m]}}(\bfx_{[2,n]}) + \omega_{\bfy}(\bfx_{[2,n]}) & y_1 = x_{1} \\
     \omega_{\bfy}(\bfx_{[2,n]})& y_1 \neq x_{1}. \\ 
  \end{cases}  
\end{align}
Therefore, $\omega_{\bfy'}(\bfx') = \omega_{\bfy}(\bfx) + \omega_{\bfy'}(\bfx)$. Using \eqref{DP-EW} recursively, it follows that $\omega_{\bfy'}(\bfx) = \sum_{i=1}^{k}\omega_{\bfy}(\bfx_{[i+1,n]})\cdot\mathbb{I}_{\alpha = x_i}$.
\end{proof}

For $\bfy \in \Sigma_2^{m}$ with $\cR\cL(\bfy)=(r_1,\ldots, r_R)$, let $x_\beta^{\bfy}, x_\gamma^{\bfy}, x_\delta^{\bfy} \in I_2(\bfy)$ be such that $\cR\cL(\bfx_\beta^\bfy)\triangleq(1,1,r_1,r_2,\ldots, r_R)$, $\cR\cL(\bfx_\gamma^\bfy)\triangleq(2,r_1,r_2,\ldots, r_R)$ and $\cR\cL(\bfx_\delta^\bfy)\triangleq(1,r_1+1,r_2,\ldots, r_R)$. The next two corollaries follow from Lemma~\ref{correction} and Lemma~\ref{alpha}.
\begin{corollary}\label{exc}
 Let $\bfy \in \Sigma_2^{m}$ it holds that: 
 \begin{enumerate}
 \item $\omega_{\bfy}(\bfx_\beta^{\bfy}) \hspace{-.2ex}=\hspace{-.2ex} 1\hspace{-.2ex}+\hspace{-.2ex}\sum_{j = 1}^{f_0} r_j$, which is maximal for $f_0 = \rho(\bfy)$.
 \item $\omega_{\bfy}(\bfx_\gamma^{\bfy}) \hspace{-.2ex}=\hspace{-.2ex} 1$. 
 \item $\omega_{\bfy}(\bfx_\delta^{\bfy}) \hspace{-.2ex}=\hspace{-.2ex} r_1+1$,  which is maximal for $\rho(\bfy) = 1$ and minimal for $\rho(\bfy) = m$.
 \end{enumerate}
\end{corollary}
\begin{corollary}\label{app1}
Let $\bfy \in \Sigma_2^{m}$, $\bfx \in \Sigma_2^{m+2}$, and $\bfy' = \alpha\circ\bfy$, $\bfx' = \alpha \circ \bfx$. If $ \alpha \ne y_{1}$ or $y_1 \ne x_{1}$, then %Let $\cR\cL(\bfy)=(r_1,\ldots, r_R)$ then we have that 
$$ {\omega_{\bfy'}(\bfx')} = \begin{cases} 
     {\omega_{\bfy}(\bfx)} + \mathbb{I}_{\bfx = \bfx_\beta^\bfy} & \alpha \neq y_{1}, y_1 = x_{1} \\
     2 \cdot {\omega_{\bfy}(\bfx)} + \mathbb{I}_{\bfx = \bfx_\gamma^\bfy} & \alpha \neq y_{1}, y_1 \neq x_{1} \\ 
     {\omega_{\bfy}(\bfx)} + \mathbb{I}_{\bfx = \bfx_\delta^\bfy} & \alpha = y_{1}, y_1 \neq x_{1}.
  \end{cases}$$
\end{corollary}
Let $\sigma^\ell \in \Sigma_{2}^\ell$ be a constant word of length $\ell$, i.e., $\rho(\sigma^\ell) = 1$. The next lemma shows that $\bfy = \sigma^m$ minimizes $\mathsf{H}^{\mathsf{In}}_{2\textrm{-}\mathsf{Del}} (\bfy)$. 
\begin{lemma}\label{min-2del}
$\mathsf{H}^{\mathsf{In}}_{2\textrm{-}\mathsf{Del}} (\bfy)$ is minimized only by $\sigma^m$ and %the minimum is given by, % It holds that, 
 
$$
\mathsf{min}^{\mathsf{In}}_{2\textrm{-}\mathsf{Del}}(n) =\mathsf{H}^{\mathsf{In}}_{2\textrm{-}\mathsf{Del}} (\sigma^m) = 2 + \dfrac{3}{4}\log\binom{m}{2}-\dfrac{1}{2}\log(m+1).
$$
 %and is attained only by $\sigma^m$.
\end{lemma}
\begin{proof} We prove the lemma by induction on $m$. Let
\begin{align}\label{ind1}
    \bfg_m \triangleq \underset{\bfy \in \Sigma_2^{m}}{\argmax}~ \mathsf{W}_{y_1\circ I_2(\bfy_{[2:m]})}(\bfy).
\end{align}
For the base case it can be verified that $\bfg_2 = \{00,11 \}$. For $\ell<m$, assume that $\bfg_{\ell-1} = \sigma^{\ell-1}$. %By Corollary \ref{cor:min1del}, $$\underset{\bfy \in \Sigma^2^{\ell-1}}{\argmax}~\mathsf{W}_{I_1(\bfy)}(\bfy)= \sigma^{\ell-1},$$ 
From \eqref{DP-EW},  ${\omega_{\bfy}(\overline{y}_1\circ\bfx)} = {\omega_{\bfy}(\bfx)},$ and by Corollary~\ref{cor:min1del} it can be deduced that
\begin{align}\label{single-del}
    \underset{\bfy \in \Sigma_2^{\ell-1}}{\argmax}~\mathsf{W}_{\overline{y}_1\circ I_1(\bfy)}(\bfy)=\sigma^{\ell-1}.
\end{align}
We will now show that $\bfg_{\ell} = \sigma^{\ell}$. Let $\bfy \in \Sigma_2^{\ell}, \cR\cL(\bfy_{[2,\ell]})\triangleq (r_1,r_2,\ldots,r_R)$, and $\bfx \in  I_2(\bfy_{[2,\ell]})$. Consider the next two cases:
\\
\textbf{Case I} - $y_1 \neq y_2$: From Corollary \ref{app1}, we know that 
$$ {\omega_{\bfy}(y_1\circ\bfx)} = \begin{cases} 
     {\omega_{\bfy_{[2,\ell]}}(\bfx)} + \mathbb{I}_{\bfx = \bfx_\beta^{\bfy_{[2,\ell]}}}  &  y_{2} = x_{1} \\
     2\cdot {\omega_{\bfy_{[2,\ell]}}(\bfx)} + \mathbb{I}_{\bfx = \bfx_\gamma^{\bfy_{[2,\ell]}}}&  y_{2} \neq x_{1}.
\end{cases}$$
Since $\bfg_{\ell-1} = \sigma^{\ell-1}$ and by Corollary \ref{exc}, it can be verified that 
$$\underset{y_1\neq y_2, \bfy_{[2:\ell]}\in \Sigma_2^{\ell-1}}{\argmax}~\mathsf{W}_{y_1\circ y_2 \circ I_2(\bfy_{[3,\ell]})}(\bfy) = \overline{\sigma^\ell_1}\circ\sigma^{\ell-1}.$$
Similarly, using Corollary \ref{exc} and \eqref{single-del}, it can be verified that
$$\underset{y_1\neq y_2, \bfy_{[2:\ell]}\in \Sigma_2^{\ell-1}}{\argmax}~\mathsf{W}_{y_1\circ \overline{y_2}\circ I_1(\bfy_{[2,\ell]})}(\bfy) = \overline{\sigma^\ell_1}\circ\sigma^{\ell-1}.$$
Therefore, 
$$\underset{y_1\neq y_2, \bfy_{[2:\ell]}\in \Sigma_2^{\ell-1}}{\argmax}~\mathsf{W}_{y_1\circ I_2(\bfy_{[2,\ell]})}(\bfy) = \overline{\sigma^\ell_1}\circ\sigma^{\ell-1}.$$
\textbf{Case II} - $y_1 = y_2$: By Corollary \ref{app1}, if  $y_{2} \neq x_{1}$, then $\omega_{\bfy}(y_1\circ\bfx) = \omega_{\bfy_{[2,\ell]}}(\bfx) + \mathbb{I}_{\bfx = \bfx_\delta^{\bfy_{[2,\ell]}}}$. As in the previous case, using Corollary \ref{exc} and \eqref{ind2}, it can be similarly verified that 
$$\underset{y_1 = y_2, \bfy_{[2:\ell]}\in \Sigma_2^{\ell-1}}{\argmax}~\mathsf{W}_{y_1\circ \overline{y}_2\circ I_1(\bfy_{[2,\ell]})}(\bfy) = \sigma^\ell.$$
Otherwise, if $y_{2} = x_{1}$, we define $\omega_{\bfy_{[2,\ell]}}(\bfx) \triangleq f(r_1,r_2,\ldots,r_R,\bfx)$, then $\omega_{\bfy}(y_1\circ\bfx) = f(r_1+1,r_2,\ldots,r_R,\bfx)$. Let
$$\cW(\bfy_{[2,\ell]}) \triangleq \mathsf{W}_{y_1\circ y_2\circ I_2(\bfy_{[3:\ell]})}(\bfy) - \mathsf{W}_{ y_2\circ I_2(\bfy_{[3:\ell]})}(\bfy_{[2,\ell]}).$$
Therefore, from \cite{ABC18} and Lemma \ref{correction}, 
\begin{align*}
    \cW(\bfy_{[2,\ell]}) =& \binom{r_1+3}{2}\log\left(\binom{r_1+3}{2}\right) - \binom{r_1+2}{2}\log\left(\binom{r_1+2}{2}\right) + \sum_{i=1}^{f_1}r_{i}\log\left(\sum_{i=1}^{f_1}r_{i}\right)\\
    &+ (\ell-1-R+1)(r_1+2)\log\left(r_1+2\right) - (\ell-1-R+1)(r_1+1)\log\left(r_1+1\right) \\
    &+ \sum_{i=2}^{R} \Big\{ (r_1+2)(r_i+1)\log\left((r_1+2)(r_i+1)\right) - (r_1+1)(r_i+1)\log\left((r_1+1)(r_i+1)\right) \Big\}.  
\end{align*}
In the Appendix, it is shown that $\underset{\bfy_{[2,\ell]} \in \Sigma_2^{\ell-1}}{\argmax}~\cW(\bfy_{[2,\ell]}) = \sigma^{\ell-1}$. Since $\bfg_{\ell-1} = \sigma^{\ell-1}$, it follows that
$$\underset{y_1 = y_2, \bfy_{[2:\ell]}\in \Sigma_2^{\ell-1}}{\argmax}~\mathsf{W}_{y_1\circ y_2\circ I_2(\bfy_{[3:\ell]})}(\bfy) = \sigma^{\ell}.$$
Therefore, 
$$\underset{y_1 = y_2, \bfy_{[2:\ell]}\in \Sigma_2^{\ell-1}}{\argmax}~\mathsf{W}_{y_1\circ I_2(\bfy_{[2:\ell]})}(\bfy) = \sigma^{\ell}.$$
The following result can be obtained by a manual comparison of these two cases,
$$\underset{\bfy \in \Sigma_2^{\ell}}{\argmax}~\mathsf{W}_{y_1\circ I_2(\bfy_{[2,\ell]})}(\bfy)=\sigma^\ell.$$
From \eqref{DP-EW},  ${\omega_{\bfy}(\overline{y}_1\circ\bfx)} = {\omega_{\bfy}(\bfx)},$ and by Corollary~\ref{cor:min1del} it can be deduced that
\begin{align}
    &\underset{\bfy \in \Sigma_2^{\ell}}{\argmax}~\mathsf{W}_{\overline{y_1}\circ I_1(\bfy)}(\bfy)=\sigma^{\ell}.\label{ind2}
\end{align}
Therefore,
$$\underset{\bfy \in \Sigma_2^{\ell}}{\argmax}~\mathsf{W}_{I_2(\bfy)}(\bfy)=\underset{\bfy \in \Sigma_2^{\ell}}{\argmax}\left\{\mathsf{W}_{y_1\circ I_2(\bfy_{[2,\ell]})}(\bfy) + \mathsf{W}_{\overline{y}_1\circ I_1(\bfy)}(\bfy)\right\}=\sigma^\ell.$$
\end{proof}

 \begin{figure*}[!ht]
\center
\includegraphics[width=.85\linewidth, height =20ex]{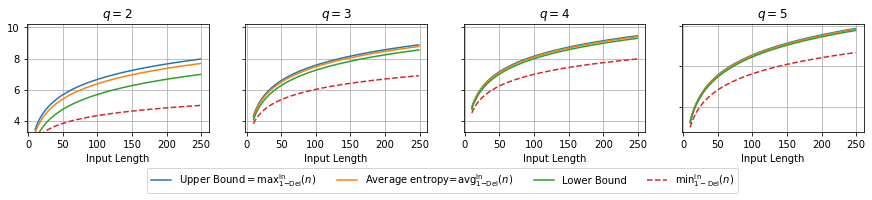}
 \caption{The minimum, maximum, average values and a lower bound on the average of the input entropy of the 1-deletion channel. } \label{fig:bounds}
\end{figure*}
\section{The Average Input Entropy}\label{sec:ave}
Let $N^m_r$ be the number of runs of length $r$ over all the sequences in $\Sigma_{q}^{m}$. It can be verified that
\begin{align*}
N^m_r = (q-1)q^{m-r-1}\left((q-1)(m-r+1)+2 \right).
\end{align*}
We now derive the expressions for the average input entropies for the $1$-deletion channel and the $1$-insertion channel. %The proof for the insertion case are omitted due to simi

%\subsection{One-Deletion Channel}
%We first characterize the average entropy and then show that it is of the order $\log(n)$. 
\begin{lemma} \label{lm:avg}
For any $n, q$, it holds that 
\begin{align*}
&\hspace{-.80ex}\mathsf{avg}^{\mathsf{In}}_{1\textrm{-}\mathsf{Del}}(n) \hspace{-0.3ex}=\hspace{-0.3ex} \log(nq) \hspace{-0.3ex}-\hspace{-0.3ex} \frac{q-1}{n}\hspace{-0.3ex}\hspace{-0.3ex}\displaystyle\hspace{-0.3ex}\sum_{r=1}^{n-1}\hspace{-0.4ex}\frac{\left(\hspace{-0.3ex}(q\hspace{-0.3ex}-\hspace{-0.3ex}1)(n\hspace{-0.3ex}-\hspace{-0.3ex}r)\hspace{-0.3ex}+\hspace{-0.3ex}2\right)\hspace{-0.5ex}(r\hspace{-0.3ex}+\hspace{-0.3ex}1)\hspace{-0.3ex}\log(r\hspace{-0.3ex}+\hspace{-0.3ex}1)}{q^{r+2}},\\
&\hspace{-.80ex}\mathsf{avg}^{\mathsf{In}}_{1\textrm{-}\mathsf{Ins}}(n) \hspace{-0.3ex}=\hspace{-0.3ex} \log(n\hspace{-0.3ex}+\hspace{-0.3ex}1) \hspace{-0.3ex}-\hspace{-0.3ex} \frac{q-1}{n+1}\displaystyle\sum_{r=1}^{n+1}\frac{\left((q\hspace{-0.3ex}-\hspace{-0.3ex}1)(n\hspace{-0.3ex}-\hspace{-0.3ex}r\hspace{-0.3ex}+\hspace{-0.3ex}2)\hspace{-0.3ex}+\hspace{-0.3ex}2 \right)r\log(r)}{q^{r+1}}.
\end{align*}
\end{lemma}
 \begin{proof}
The following equalities hold
\begin{align*}
\mathsf{avg}^{\mathsf{In}}_{1\textrm{-}\mathsf{Del}}(n) %= \mathbb{E}_{\Sigma_{q}^{n-1}}[ \mathsf{H}^{\mathsf{In}}_{1\textrm{-}\mathsf{Del}} (\bfy)] \\
&=\log(nq) - \frac{1}{nq}\mathbb{E}_{\Sigma_{q}^{n-1}}\left[\sum_{i=1}^{\rho(\bfy)}(r_{i}+1)\log(r_{i}+1)\right]\\
&=\log(nq) - \frac{1}{nq^n}\sum_{\bfy\in\Sigma_{q}^{n-1}}\sum_{i=1}^{\rho(\bfy)}(r_{i}+1)\log(r_{i}+1) \\
&= \log(nq) - \frac{1}{nq^{n}}\sum_{r=1}^{n-1}N_r(r+1)\log(r+1)\\
% &= \log(nq) 
% - \frac{q-1}{nq^{n}}\left(\displaystyle\sum_{r=1}^{n-1}\frac{q^{n}}{q^{r+2}}\left((q-1)(n-r)+2 \right)(r+1)\log(r+1)\right)\\
&= \log(nq) - \frac{q-1}{n}\displaystyle\sum_{r=1}^{n-1}\frac{\left((q-1)(n-r)+2 \right)(r+1)\log(r+1)}{q^{r+2}}.
\end{align*}
The value of $\mathsf{avg}^{\mathsf{In}}_{1\textrm{-}\mathsf{Ins}}(n)$ can be evaluated similarly.
\end{proof}

Clearly, $\mathsf{avg}^{\mathsf{In}}_{1\textrm{-}\mathsf{Del}}(n)\le \mathsf{max}^{\mathsf{In}}_{1\textrm{-}\mathsf{Del}}(n)$ and $\mathsf{avg}^{\mathsf{In}}_{1\textrm{-}\mathsf{Ins}}(1)\le \mathsf{max}^{\mathsf{In}}_{1\textrm{-}\mathsf{Ins}}(n)$ for any integers $n,q$. To improve the lower bounds on the average entropies which are better than the minimum values, note that if $r\geq1$, then $r\geq\log(r+1)$. Hence by Lemma~\ref{lm:avg}, it can be shown that,  
\begin{align*}
\mathsf{avg}^{\mathsf{In}}_{1\textrm{-}\mathsf{Del}}(n) \hspace{-.5ex}&\geq  \log(nq) - \frac{1}{n}\left(\frac{2n}{q-1}-\frac{n^2-n}{q^{n+1}}+\frac{2q^2-2q^{n+2}}{(q-1)^2 q^{n+1}}\right)
\hspace{-.5ex}
\end{align*}
and 
\begin{align*}
\mathsf{avg}^{\mathsf{In}}_{1\textrm{-}\mathsf{Ins}}(n) \hspace{-.5ex}&\geq  \hspace{-.5ex}\log(nq) \hspace{-.5ex} + \hspace{-.5ex} \frac{1}{n+1} \hspace{-.5ex} \left(\hspace{-.7ex} \frac{n^2}{q^{n+2}}\hspace{-.5ex}-\hspace{-.5ex}\frac{n(2q^{n+2}\hspace{-.5ex}-\hspace{-.5ex}q\hspace{-.5ex}+\hspace{-.5ex}1)}{(q\hspace{-.5ex}-\hspace{-.5ex}1)q^{n+2}}\hspace{-.5ex}+\hspace{-.5ex}\frac{2(q^n\hspace{-.5ex}-\hspace{-.5ex}1)}{(q\hspace{-.5ex}-\hspace{-.5ex}1)^2q^n}\hspace{-.5ex}\right)\hspace{-.5ex}.
\hspace{-.5ex}
\end{align*}
The bounds for $\mathsf{avg}^{\mathsf{In}}_{1\textrm{-}\mathsf{Del}}(n)$ are presented in Figure~\ref{fig:bounds}.

\section*{Appendix}
Let $\bfy \in \Sigma_2^{m}, \cR\cL(\bfy)\triangleq (r_1,r_2,\ldots,r_R)$. Then,
\begin{align*}
    \cW(\bfy) &= \binom{r_1+3}{2}\log\binom{r_1+3}{2} - \binom{r_1+2}{2}\log\binom{r_1+2}{2} + \left(1+\sum_{i=1}^{f_1}r_{i}\right)\log\left(1+\sum_{i=1}^{f_1}r_{i}\right)  - \left(\sum_{i=1}^{f_1}r_{i}\right)\log\left(\sum_{i=1}^{f_1}r_{i}\right)\\
    &+ (m-R+1)(r_1+2)\log\left(r_1+2\right) - (m-R+1)(r_1+1)\log\left(r_1+1\right)\\
    &+ \sum_{i=2}^{R} \Big\{ (r_1+2)(r_i+1)\log\left((r_1+2)(r_i+1)\right) - (r_1+1)(r_i+1)\log\left((r_1+1)(r_i+1)\right) \Big\}.
\end{align*}
We know that $(r + 1)\log(r + 1) - r\log(r)$ is an increasing function w.r.t. $r$ and it can be verified that $\sum_{i=1}^{f_1}r_{i}\log\left(\sum_{i=1}^{f_1}r_{i}\right)$ is maximized when $f_1=R$. Therefore, 
$$\underset{\bfy \in \Sigma_2^{m}}{\argmax}~\left\{\left(1+\sum_{i=1}^{f_1}r_{i}\right)\log\left(1+\sum_{i=1}^{f_1}r_{i}\right)  - \left(\sum_{i=1}^{f_1}r_{i}\right)\log\left(\sum_{i=1}^{f_1}r_{i}\right)\right\} = \sigma^m.$$
Next, let 
\begin{align*}
    \cW'(\bfy) = & (m-R+1)(r_1+2)\log\left(r_1+2\right) - (m-R+1)(r_1+1)\log\left(r_1+1\right) \\
    &+ \sum_{i=2}^{R} \Big\{ (r_1+2)(r_i+1)\log\left((r_1+2)(r_i+1)\right) - (r_1+1)(r_i+1)\log\left((r_1+1)(r_i+1)\right) \Big\},
\end{align*}
and assume to the contrary that for the sequence $\bfy'\in\Sigma_2^m$ with $\cR\cL(\bfy)\triangleq (r_1+1,r_2,\ldots, r_j-1, \ldots, r_R)$, the following holds,
$$
\cW'(\bfy)+ \binom{r_1+3}{2}\log\binom{r_1+3}{2} - \binom{r_1+2}{2}\log\binom{r_1+2}{2} > \cW'(\bfy') + \binom{r_1+4}{2}\log\binom{r_1+4}{2} - \binom{r_1+3}{2}\log\binom{r_1+3}{2}.
$$

First consider the difference, 
\begin{align*}
\cW'(\bfy') - \cW'(\bfy) =& (m-R+1)(r_1+3)\log\left(r_1+3\right) - (m-R+1)(r_1+2)\log\left(r_1+2\right) \\
&+ \sum_{i=2}^{R} \Big\{ (r_1+3)(r_i+1)\log\left((r_1+3)(r_i+1)\right) - (r_1+2)(r_i+1)\log\left((r_1+2)(r_i+1)\right) \Big\} \\ 
&- \Big\{ (r_1+3)(r_j+1)\log\left((r_1+3)(r_j+1)\right) - (r_1+2)(r_j+1)\log\left((r_1+2)(r_j+1)\right) \Big\} \\ 
&+ \Big\{ (r_1+3)(r_j)\log\left((r_1+3)(r_j)\right) - (r_1+2)(r_j)\log\left((r_1+2)(r_j)\right) \Big\} \\ 
&-\Bigg\{(m-R+1)(r_1+2)\log\left(r_1+2\right) - (m-R+1)(r_1+1)\log\left(r_1+1\right) \\
&+ \sum_{i=2}^{R} \Big\{ (r_1+2)(r_i+1)\log\left((r_1+2)(r_i+1)\right) - (r_1+1)(r_i+1)\log\left((r_1+1)(r_i+1)\right) \Big\} \Bigg\}.
\end{align*}

After some simplification it follows that
\begin{align*}
\cW'(\bfy') - \cW'(\bfy) = &  (m-R+1)\bigg\{(r_1+3)\log\left(r_1+3\right) - 2(r_1+2)\log\left(r_1+2\right) + (r_1+1)\log\left(r_1+1\right)\bigg\} \\
&+ \sum_{i=2}^{R} (r_i+1)\Big\{ (r_1+3)\log\left(r_1+3\right) - 2(r_1+2)\log\left(r_1+2\right) + (r_1+1)\log\left(r_1+1\right) \Big\} \\ 
&+ (r_1+2)\log(r_1+2) - (r_1+3)\log(r_1+3) + (r_j)\log(r_j) - (r_j+1)\log(r_j+1) \\
=& (2m-r_1)\bigg\{(r_1+3)\log\left(r_1+3\right) - 2(r_1+2)\log\left(r_1+2\right) + (r_1+1)\log\left(r_1+1\right)\bigg\} \\
&+ (r_1+2)\log(r_1+2) - (r_1+3)\log(r_1+3) + (r_j)\log(r_j) - (r_j+1)\log(r_j+1). \\
\intertext{We know that $(r + 1)\log(r + 1) - r\log(r)$ is an increasing function w.r.t. $r$. Therefore,}
\cW'(\bfy') - \cW'(\bfy)&\geq(2m-r_1)\bigg\{(r_1+3)\log\left(r_1+3\right) - 2(r_1+2)\log\left(r_1+2\right) + (r_1+1)\log\left(r_1+1\right)\bigg\} \\
&+ (r_1+2)\log(r_1+2) - (r_1+3)\log(r_1+3) + (m)\log(m) - (m+1)\log(m+1).
\end{align*}
Consider now the difference,
\begin{align*}
&\cW'(\bfy') - \cW'(\bfy)+ \binom{r_1+4}{2}\log\binom{r_1+4}{2} - 2\binom{r_1+3}{2}\log\binom{r_1+3}{2} +\binom{r_1+2}{2}\log\binom{r_1+2}{2} \\
\geq&  (2m-r_1)\bigg\{(r_1+3)\log\left(r_1+3\right) - 2(r_1+2)\log\left(r_1+2\right) + (r_1+1)\log\left(r_1+1\right)\bigg\} \\
&  + (r_1+2)\log(r_1+2) - (r_1+3)\log(r_1+3) + (m)\log(m) - (m+1)\log(m+1)\\
&  + \binom{r_1+4}{2}\log\binom{r_1+4}{2} - 2\binom{r_1+3}{2}\log\binom{r_1+3}{2} +\binom{r_1+2}{2}\log\binom{r_1+2}{2}.
\end{align*}
It can be verified that the above lower bound is a decreasing function w.r.t. $r_1$ and that at $r_1 = m$ it is greater than $0$. Therefore this a contradiction to our assumption. Hence, 
$$
\underset{\bfy \in \Sigma_2^{m}}{\argmax}~\left\{
\cW'(\bfy)+ \binom{r_1+3}{2}\log\binom{r_1+3}{2} - \binom{r_1+2}{2}\log\binom{r_1+2}{2}\right\} = \sigma^m,
$$
and thus $\underset{\bfy \in \Sigma_2^{m}}{\argmax}~\cW(\bfy) = \sigma^m$. 
\end{document}